\documentclass[10pt]{article}

\usepackage{amssymb, amsmath, amsthm}
\usepackage{natbib}

\newtheorem{thm}{Theorem}[section]
\newtheorem{prop}[thm]{Proposition}

\newtheorem{ex}[thm]{Example}

\newcommand{\RR}{\mathbb{R}}
\newcommand{\NN}{\mathbb{N}}

\newcommand{\EE}{\mathbb{E}}
\newcommand{\PP}{\mathbb{P}}

\usepackage{dsfont}
\usepackage{subcaption,graphicx}

\RequirePackage[%
reversemp,
paperwidth=155mm,
paperheight=235mm,
top={19.5mm},
headheight={5.5pt},
headsep={5.6mm},
text={119mm,194.25mm},
marginparsep=5mm,
marginparwidth=12mm,
bindingoffset=6mm,
footskip=10mm]{geometry}
\begin{document}
	
\title{Non-stationary max-stable models with an application to heavy rainfall
data}

\author{Carolin Forster$^1$ and Marco Oesting$^{1,2}$}

\date{%
	$^1$ {\normalsize  Institute for Stochastics and Applications, University of Stuttgart,
	70569 Stuttgart, Germany}\\%
	$^2$ {\normalsize Stuttgart Center for Simulation Science (SC SimTech), University of Stuttgart,
	70569 Stuttgart} \\[2ex]%
	\today
}

%\author*[1]{\fnm{Carolin} \sur{Forster}}\email{carolin.forster@mathematik.uni-stuttgart.de}

%\author[1,2]{\fnm{Marco} \sur{Oesting}}\email{marco.oesting@mathematik.uni-stuttgart.de}

%\affil[1]{\orgdiv{Institute for Stochastics and Applications}, \orgname{University of Stuttgart}, \orgaddress{\street{Allmandring 5b}, \postcode{70569} \city{Stuttgart}, \country{Germany}}}

%\affil[2]{\orgdiv{Stuttgart Center for Simulation Science}, \orgname{University of Stuttgart}, \orgaddress{\street{Allmandring 5b}, \postcode{70569} \city{Stuttgart}, \country{Germany}}}

\maketitle

\begin{abstract}
	In recent years, parametric models for max-stable processes have become a popular choice for modeling spatial extremes because they arise as the asymptotic limit of rescaled maxima of independent and identically distributed random processes. Apart from few exceptions for the class of extremal-$t$ processes, existing literature mainly focuses on models with stationary dependence structures. In this paper, we propose a novel non-stationary approach that can be used for both Brown-Resnick and extremal-$t$ processes -- two of the most popular classes of max-stable processes -- by including covariates in the corresponding variogram and correlation functions, respectively. 
We apply our new approach to extreme precipitation data in two regions in Southern and Northern Germany and compare the results to existing stationary models in terms of Takeuchi’s information criterion (TIC). Our results indicate that, for this case study, non-stationary models are more appropriate than stationary ones for the region in Southern Germany. 
In addition, we investigate theoretical properties of max-stable processes conditional on random covariates. We show that these can result in both asymptotically dependent and asymptotically independent processes. Thus, conditional models are more flexible than classical max-stable models. 
\end{abstract}

%\keywords{Max-Stable Process, Non-Stationarity, Covariates, Asymptotic Independence}

\section{Introduction}\label{sec1}

Weather extremes such as heavy rainfall often cause enormous social and economic damages.
For example, the summer 2021 flood event in Germany has involved at least 40,000 people in the Ahr valley and nearby regions of North-Rhine Westphalia \citep{bosseler2021living}. 
It is anticipated that an increase in extreme precipitation will occur due to global warming \citep[cf. for example][]{ipcc,asadieh2015global}. 
Often, such events occur simultaneously in various locations.
Thus, models for spatial extremes are of great interest as they can support a better understanding and knowledge about such events. In our work, we model spatial extremes
based on spatially indexed block maxima \
$\{z(s), \, s \in S\}$ for locations $s$ in some region $S \subset \mathbb{R}^2$ via max-stable processes that are the only non-degenerate limit processes of rescaled maxima.

In recent years, these processes have been used in many environmental applications, for instance with heavy rainfall \citep[see, e.g.,][]{coles1996modelling,davison2012statistical,sebille2017modeling,gelfand2019handbook}, extreme temperatures \citep[see, e.g.,][]{huser-genton-16,thibaud2016bayesian,davison2012statistical}, 
extreme wind \citep[see, e.g.,][]{Engelke2015,Genton2015,oesting2017statistical}, severe storms \citep[see, e.g.,][]{koh-etal-2022}, drought \citep[see, e.g.,][]{oesting2018spatial}, or extreme snowfall or depth \citep[see e.g.,][]{blanchetdavison2011,gaume2013mapping,nicolet2015inferring}. 
Within the class of max-stable processes, the subclasses of extremal-$t$ processes \citep{opitz2013extremal} and Brown-Resnick processes \citep{kab_2009} belong to the most frequently used parametric models. So far, the research in this field mainly focused on 
models with dependence structures that are stationary in space. 

Due to the complexity of underlying physical phenomena, however, in many applications, it is reasonable to assume that the dependence structure is not stationary, but, for instance, also depends on covariates. Ignoring these effects may describe the data inadequately and, thus, may result in erroneous estimates for return levels or other characteristics of interest. 
In recent years, there have been some first approaches to construct models with non-stationary dependence structures. 
For instance, \citet{blanchetdavison2011} achieve non-stationarity by applying classical stationary max-stable models on a higher-dimensional climate space which, additionally to the geographical coordinates, includes further suitable covariates.
In contrast, \citet{youngman} suggests to deform the domain of the considered spatial process by utilizing a spatial deformation or dimension expansion and apply a stationary model in the transformed space in order to get a non-stationary dependence function in the original domain. 
Instead of using a stationary max-stable model in a climate space, that is build parametrically, \citet{chevalier} introduce to replace the climate space by a latent space. In contrast, this latent space is constructed with nonparametric fitting approaches, that make use of multidimensional scaling (MDS).
Another example is the approach by \citet{huser-genton-16} who propose a non-stationary approach that connects max-stable processes, especially the extremal-$t$ model, with the non-stationary correlation model of \citet{paciorek-shervish-06}.
Their correlation function contains spatially varying covariance matrices, where significant covariates are included. This approach with correlation functions can also be applied to variograms, i.e., to max-stable Brown-Resnick processes with the restriction that the underlying variograms are bounded, see \citet{shao2022flexible}. 

In this paper, we propose an approach that can be used for extremal-$t$ and Brown-Resnick processes covering both the cases of bounded and unbounded variograms by applying a stationary max-stable model on a higher dimensional space including covariates. Similarly to the approaches by \citet{blanchetdavison2011}, \citet{youngman} and \citet{chevalier}, our approach is based on the idea to use a stationary approach in a different space. Our approach is closest to the first of the three approaches mentioned above, i.e., the approach of \citet{blanchetdavison2011}, which they use for the Schlather model \citep{schlather2002models}, i.e., a special case of the extremal-$t$ model that we consider here. Furthermore, we consider more general types of correlation functions.

The outline of this paper is as follows: In Section \ref{Theoretical background}, we provide the theoretical background on max-stable processes particularly focusing on the classes of extremal-$t$ and Brown-Resnick processes. Besides stationary approaches, the section also covers existing non-stationary approaches and spatial dependence measures. Section \ref{Novel approach for Non-Stationarity} is dedicated to our novel non-stationary approach for both Brown-Resnick and extremal-$t$ processes, where we include covariates in the corresponding variogram and correlation functions, respectively. In Section \ref{Application}, we apply these models to heavy rainfall data in Germany
and compare the results to existing models. 
While the former sections focus on max-stable models with fixed covariates, in Section \ref{bridgebetweenasympandsy}, we provide an extension to the case of random covariate processes. In particular, we study extremal dependence in random scale constructions of conditional max-stable models resulting both in asymptically dependent and asymptotically independent processes. Finally, Section \ref{Discussion} concludes with a discussion. 

\section{Theoretical background}
\label{Theoretical background}

\subsection{Max-stable processes}
Let $\mathcal{F} \subseteq \{ f: \mathbb{R}^d \to \mathbb{R} \}$ be some suitable space of real-valued functions on $\RR^d$ equipped with the $\sigma$-algebra generated by the cylinder sets of the form 
\begin{equation}
\label{cylinder_sets}
 \{ f \in \mathcal{F}: \, f(s_1) \in B_1, \ldots, f(s_n) \in B_n \}   
\end{equation}
where $s_1,\ldots,s_n \in \RR^d$, $B_1,\ldots,B_n \subset \RR$ are Borel sets and $n \in \NN$. Furthermore, let $X=\{X(s), \, s \in \mathbb{R}^d\}$ be a stochastic process, i.e., a $\mathcal{F}$-valued random object and denote independent copies of $X$ by $X_{1},X_{2},\ldots$. If there exist functions $a_{n}: \RR^d \to (0,\infty)$ and $b_{n}: \RR^d \to \mathbb{R}$ such that
\begin{equation} \label{maxstable}
\mathcal{L} \left\{\frac{\max_{i=1}^n X_i(s) - b_n(s)}{a_n(s)}, \, s \in \mathbb{R}^d \right\} \stackrel{n \to \infty}{\longrightarrow} \mathcal{L}\{Z(s), \, s \in \mathbb{R}^d \}
\end{equation}
weakly in $\mathcal{F}$
and the univariate margins of $X=\{X(s), \, s \in \mathbb{R}^d\}$ are non-degenerate, then the limit process $Z$ is necessarily a max-stable process, i.e., it satisfies
the following max-stability property: for all $s \in \mathbb{R}^{d}$, there exist sequences $\{a'_{n}(s)\}_{n \in \mathbb{N}} \subset (0,\infty)$ and $\{b'_{n}(s)\}_{n \in \mathbb{N}} \subset \mathbb{R}$ s.t.
\begin{equation*} 
\mathcal{L} \left\{\frac{\max_{i=1}^n Z_i(s) - b_n'(s)}{a_n'(s)}, \, s \in \mathbb{R}^d \right\}\stackrel{d}= \mathcal{L}\{Z(s), \, s \in \mathbb{R}^d\},
\end{equation*}
where $Z_i$, $i \in \mathbb{N}$, are independent copies of $Z$.
As the law of the process $Z$ is uniquely defined by its finite dimensional distributions, this is equivalent to
\begin{align*}
 &{}\PP\left( Z(s_{1}) \leq a'_{n}(s_{1})z_{1}+ b'_{n}(s_{1}),\cdots,Z(s_{k}) \leq a'_{n}(s_{k})z_{k}+b'_{n}(s_{k})
        \right)^{n}\\
&= \PP\left( Z(s_{1}) \leq z_{1},\cdots,Z(s_{k}) \leq z_{k})
        \right),
\end{align*}  
for all $z_1,\hdots,z_k \in \mathbb{R}, s_1,\hdots,s_k \in \mathbb{R}^{d}$ and $k \in \mathbb{N}$ \citep[see for instance,][]{huser2020advances}.

From univariate extreme value theory \citep[see, for instance,][]{Embrechts1997,coles2001introduction}, it is known that the non-degenerate margins of the max-stable process $Z$ in \eqref{maxstable} follow a generalized extreme value (GEV) distribution, which can be described via the parameters $\xi \in \mathbb{R}$ (shape), $\mu \in \mathbb{R}$ (location) and $\sigma > 0$ (scale) by the following cumulative distribution function:
\begin{equation*}
 G_{\xi, \mu, \sigma}(x) = \begin{cases} 
     \exp\Big(-\left(1 + \xi \frac{x - \mu}{\sigma}\right)^{-1/\xi}\Big), & \xi \neq 0\\
     \exp\left(-\exp\left(-\frac{x-\mu}\sigma\right)\right), & \xi=0   
                           \end{cases},
   \quad 1 + \xi \frac{x - \mu}{\sigma} > 0.
\end{equation*}
Under marginal transformations between GEV distributions, the max-stability property is maintained.
Thus, it is a common choice to consider max-stable processes on the unit Fr\'echet scale, i.e.,
\begin{equation*}
  \PP(Z(s) \leq z)=\exp\left(-\frac{1}{z}\right), \quad z>0, \quad s \in \mathbb{R}^d. 
\end{equation*}
These processes are called simple max-stable processes.
By \citet{de1984spectral}, any simple max-stable process in $\mathcal{F}=C(\RR^d)$, i.e., any sample-continuous simple max-stable process can be constructed by 
\begin{equation}
\label{spectralrepresentation}
Z(s) = \max_{i \in \NN}  U_i W_i(s), \quad s \in \mathbb{R}^d,  
\end{equation}
where $U_i$ are points of a Poisson point process
on $(0,\infty)$ with intensity $u^{-2} \mathrm{d}u$ and $W_i$ ($i \in \NN$) are independent copies of a nonnegative sample-continuous stochastic process $W$ on $\mathbb{R}^d$, 
called spectral process,  with $\EE(W(s))=1$, $s \in \mathbb{R}^d$.
According to the construction in \eqref{spectralrepresentation}, max-stable processes can be interpreted as the pointwise maxima of ``storms" with amplitudes $U_i$ and shapes $W_i(\cdot)$
\citep[cf.][]{smith1990max}.
If $W$ is not sample-continuous, the resulting process $Z$ in \eqref{spectralrepresentation} is no longer sample-continuous, but still max-stable w.r.t.\ all finite dimensional distributions, i.e., max-stable with respect to the space $\mathcal{F} = \{f: \RR^d \to \RR\}$.

From representation \eqref{spectralrepresentation}, it also follows that
the joint cumulative distribution function of any random vector $(Z(s_{1}),\hdots,Z(s_{k}))^{T}$, $s_1,\hdots,s_k \in \mathbb{R}^d$, can be written as
\begin{align}
\label{jdfrv}
\PP(Z(s_{1}) \leq z_{1},\hdots,Z(s_{k}) \leq z_{k})&=\exp(-V_{s_1,\hdots,s_k}(z_{1},\hdots,z_{k})),
\end{align}
for $z_1,\hdots,z_k >0$, where 
$$V_{s_1,\hdots,s_k}(z_{1},\hdots,z_{k})=\EE\left[\max\left\{\frac{W(s_{1})}{z_{1}},\hdots,\frac{W(s_{k})}{z_{k}}\right\}\right]$$
is the so-called exponent function.

Thus, if $V$ is differentiable, the corresponding joint probability density function has the form 
\begin{equation}
\label{density}
\begin{aligned}
&f_{s_{1},\hdots,s_{k}}(z_{1},\hdots,z_{k}) \\
&=\exp\{-V_{s_1,\hdots,s_k}(z_{1},\hdots,z_{k})\}\sum_{\pi \in \mathcal{P}}\prod_{\tau \in \pi} \left\{-\frac{\partial}{\partial z_\tau} V_{s_1,\ldots,s_k}(z_{1},\hdots,z_{k})\right\},
\end{aligned}
\end{equation}
where $\mathcal{P}$ is the set of all partitions of $K=\{1,\hdots,k\}$ and $\frac{\partial}{\partial z_\tau} V_{s_1,\ldots,s_k}$ indicates the partial derivative of $V_{s_1,\hdots,s_k}$  w.r.t.\ to all elements of $z_\tau = (z_i)_{i \in \tau}$ \citep{huser2020advances}   \
Consequently, the number of summands in \eqref{density} equals the cardinality of $\mathcal{P}$. This superexponentially growing number, the so-called Bell number, makes the computation of the full likelihood intractable even for moderate dimensions \citep{gelfand2019handbook}. 
The most common solution is to use the pairwise likelihood which is based on bivariate probability density functions of the form 
\begin{align}
\label{bivdensity}
f_{s_{i},s_{j}}(z_{i},z_{j})={} \exp\{-V_{s_i,s_j}(z_{i},z_{j})\} 
\Big\{ & \frac{\partial}{\partial z_i}  V_{s_{i},s_{j}}(z_{i},z_{j})\frac{\partial}{\partial z_j} V_{s_{i},s_{j}}(z_{i},z_{j}) \nonumber\\ 
&  -\frac{\partial}{\partial z_i \partial z_j} V_{s_{i},s_{j}}(z_{i},z_{j})\Big\}
\end{align}

for $1 \leq i \neq j \leq k$ only.
More details on the pairwise likelihood, which uses a composition of bivariate density functions of the form \eqref{bivdensity}, can be found in Section \ref{modelfittingandcomparison}. A special case of the joint distribution function of random vectors in \eqref{jdfrv} is 
\begin{align*}
\PP(Z(s_{1}) \leq z,\hdots,Z(s_{k}) \leq z)
&= \exp(-z^{-1}V_{s_1,\hdots,s_k}(1,\hdots,1))\\
&=\exp(-\theta(s_1,\hdots,s_k)/z), 
\end{align*}
with the first equation utilizing the homogeneity of the exponent function. 
The characteristic $\theta(s_1,\hdots,s_k)=V_{s_1,\hdots,s_k}(1,\hdots, 1)$ is called extremal coefficient and serves as a measure of extremal dependence and for $k$ locations. Its value ranges between $1$ and $k$ and can be interpreted as the effective number of independent random variables among $Z(s_1),\ldots,Z(s_k)$, i.e., a value of $k$ indicates full independence, while $1$ means perfect dependence \cite[cf.][]{schlather-tawn-02}.  

In the bivariate setting, alternatively, the (upper) tail dependence coefficient, defined by 
$$ \chi(s_1,s_2) = \lim_{u \to \infty} \PP(Z(s_2) > u \mid Z(s_1) > u) $$ 
for all $ s_1, s_2 \in \RR^d$ can be used to measure extremal dependence. 
More precisely, two random variables $Z(s_1)$ and $Z(s_2)$ are called asymptotically independent if $\chi(s_1,s_2) = 0$ and they are called asymptotically dependent if $\chi(s_1,s_2) > 0$.
Asymptotic dependence indicates a positive probability that extreme events take place simultaneously at several locations independent of the threshold size.
If $Z$ is a max-stable process as above, it additionally holds
$$\chi(s_1,s_2)=2 - V_{s_1,s_2}(1,1) =2 - \theta(s_1,s_2)  \quad s_1, s_2 \in \mathbb{R}^d. $$ Then, asymptotic independence implies that the two random variables $Z(s_1)$ and $Z(s_2)$ are independent.
As will be demonstrated in Section \ref{Application}, in our application, the assumption of asymptotic dependence is reasonable. However, it depends on the application and, in many cases, environmental data exhibit a decrease in the dependency as the events become more extreme \citep[see, e.g.,][]{huser2020advances} which indicates asymptotic independence.
In the following subsections, we consider two of the most popular max-stable models. 

\subsection{Extremal-$t$ processes} \label{sec:extremalt}

The extremal-$t$ process \citep{opitz2013extremal}
is a max-stable process based on i.i.d.\ Gaussian processes $\varepsilon_{i}(\cdot)$ with mean zero and variance one. Its spectral process in \eqref{spectralrepresentation} may be written as
\begin{equation}
\label{et_Wi}
 W_{i}(s) = c_{\nu} \cdot \max\{0,\varepsilon_{i}(s)\}^{\nu}, \quad s \in \mathbb{R}^d, 
\end{equation}
with degree of freedom $\nu > 0$ and scaling factor $c_{\nu}=\pi^{\frac{1}{2}}2^{1-\frac{\nu}{2}}/\Gamma\{(\nu+1)/2\}$ that guarantees the condition $\EE(W_{i}(s))=1$, $s \in \mathbb{R}^d$. Moreover, there exists a closed formula for the bivariate exponent function that is 
\begin{equation*}
\begin{aligned}
V_{s_1,s_2}(z_{1},z_{2})
&=  \frac{1}{z_{1}} T_{\nu+1}\Big\{ -\frac{\rho(s_1,s_2)}{a(s_1,s_2)}+\frac{1}{a(s_1,s_2)}\Big(\frac{z_{2}}{z_{1}} \Big)^{\frac{1}{\nu}}\Big \}\\
&+\frac{1}{z_{2}} T_{\nu+1}\Big\{ -\frac{\rho(s_1,s_2)}{a(s_1,s_2)}+\frac{1}{a(s_1,s_2)}\Big(\frac{z_{1}}{z_{2}} \Big)^{\frac{1}{\nu}}\Big \}, \quad z_{1},z_{2} >0, 
\end{aligned}
\end{equation*}
where $T_{w}(\cdot)$ denotes the c.d.f.\ of a student-$t$ distribution with $w$ degrees of freedom, $\rho: \RR^d \times \RR^d \to [-1,1]$ defined by
$\rho(s_1,s_2)=\mathrm{corr}\{\varepsilon_{i}(s_{1}),\varepsilon_{i}(s_{2})\}$ is the correlation function of $\varepsilon_{i}(\cdot)$ and $a(s_1,s_2)=(1-\rho(s_1,s_2)^{2})/(\nu+1)$ 
\citep[cf.][]{gelfand2019handbook}. 
It can be deduced that the bivariate extremal coefficient $\theta(s_1,s_2)$ is given by 
\begin{equation*}
\theta(s_1,s_2)=V_{s_1,s_2}(1,1) =  2T_{\nu+1}\Big\{ -\frac{\rho(s_1,s_2)}{a(s_1,s_2)}+\frac{1}{a(s_1,s_2)}\Big \}.
\end{equation*}
Analogously, it can be seen that all finite-dimensional distributions and, consequently, the distribution of the process, can be fully described by the correlation function $\rho(\cdot,\cdot)$  of the underlying Gaussian field and the parameter $\nu$. In particular, the extremal-$t$ model is stationary if and only if $\rho(s_1,s_2)$ depends on $s_1-s_2$ only, i.e., $\varepsilon_{i}(\cdot)$ is stationary.

\subsection{Brown-Resnick processes}

The Brown-Resnick process \citep{kab_2009}
is also a max-stable process based on i.i.d.\ centered Gaussian processes $\tilde{\varepsilon_{i}}(\cdot)$.
 For this model, the stochastic process $W_{i}(\cdot)$ in \eqref{spectralrepresentation} has the form  
$$ W_{i}(s) = \exp\Bigg\{\tilde{\varepsilon_{i}}(s)- \frac 1 2 \mathrm{Var}(\tilde{\varepsilon_{i}}(s))\Bigg\}, \quad s \in \mathbb{R}^d,$$
and the finite-dimensional distributions of the Brown-Resnick process $Z$ depend on the function
      $$ \gamma: \RR^d \times \RR^d \to [0,\infty), \  \gamma(s_1,s_2) = \frac 1 2 \mathrm{Var}(\tilde{\varepsilon_{i}}(s_1) - \tilde{\varepsilon_{i}}(s_2)), $$
      called the variogram of $\tilde{\varepsilon_{i}}(\cdot)$, only \citep[see,][Theorem 1 with $\lambda =1$]{kabluchko2011extremes}.
The variogram of $\tilde{\varepsilon_{i}}(\cdot)$ is a conditionally negative definite function with
\begin{equation*}
\gamma(s,s)=0, \quad s \in \mathbb{R}^d, 
\end{equation*}
see \citet{kab_2009}. There exists a closed formula for the bivariate exponent function, that is, 
\begin{equation*}
V_{s_1,s_2}(z_{1},z_{2}) = \frac{1}{z_{1}}\Phi\Big\{ \frac{b}{2}-\frac{1}{b}\log\Big(\frac{z_{1}}{z_{2}} \Big)\Big\}+\frac{1}{z_{2}}\Phi\Big\{ \frac{b}{2}-\frac{1}{b}\log\Big(\frac{z_{2}}{z_{1}} \Big)\Big\},  
\end{equation*}
where $b=\sqrt{(2\gamma(s_{1},s_{2}))}$ \citep{huser2013composite}. Thus, the bivariate extremal $\theta(s_1,s_2)$ coefficient is 
\begin{equation*}
\theta(s_1,s_2)=V_{s_1,s_2}(1,1) = 2\Phi\Big\{ \frac{b}{2}\Big\}.  
\end{equation*}
From the fact that the variogram determines the distribution of the process uniquely, it can be seen that the Brown-Resnick process is stationary if and only if $\gamma(s_1,s_2)$ only depends on the separation vector $h = s_1 - s_2$ and this is the case if the Gaussian process $\tilde{\varepsilon_{i}}(\cdot)$ possesses stationary increments. 

\subsection{Stationary approaches}
The most common stationary approaches for spatial extremes use isotropic or geometric anisotropic models. In the isotropic extremal-$t$ models and Brown-Resnick models, the correlation function and variogram, respectively, solely depend on the distance of geographic coordinates. An example of such a correlation function is the powered exponential covariance function
\begin{align}
\label{iso_corr}
  &\rho(\|h\| ) = \exp\Big\{-\left(\| qh\|\right)^{\alpha } \Big\}, \quad h \in \RR^d,
	\end{align}
and an example of such a variogram is the power variogram model
\begin{align}
\label{iso_var}
  &\gamma(\|h\| ) = \left(\| qh\|\right)^{\alpha},\quad h \in \RR^d,
	\end{align}
where $1/q>0$ is the range and $\alpha \in (0,2]$ the smoothness parameter. 

In contrast to the isotropic model, the correlation function and variogram, respectively, of a stationary but anisotropic model also changes with direction. Such a model can be obtained, for instance, by building 
in an anisotropy matrix $A$ which allows for rotation and dilation.
A valid correlation function
\begin{align}
\label{ani_corr}
	&\rho(\|h\| ) = \exp\Big\{-\left(\| Ah\|\right)^{\alpha } \Big\},\\
	    \intertext{and a valid variogram}
	    \label{ani_var}
	    &\gamma(\|h\|) = \left(\| Ah\|\right)^{\alpha },
\end{align}
can be obtained for any anisotropy matrix $A$, e.g., in the case $d=2$, 
\begin{align}
    \label{animat}
		&A=\left( \begin{array}{cc} q_1& 0\\ 0 & q_2 \end{array} \right) \cdot
		\left( \begin{array}{cc} \cos(\vartheta) & \sin(\vartheta)\\ -\sin(\vartheta) 				& \cos(\vartheta) \end{array} \right),
\end{align}
with parameters $\alpha \in (0,2]$, $-\frac \pi 4 \leq \vartheta \leq \frac \pi 4$, $q_1,q_2>0$.
This specific choice allows for different range parameters $1/q_1$ and $1/q_2$ and a rotation specified by $\vartheta$.

\subsection{Non-stationary approaches}
In this section, two existing non-stationary approaches are described in more detail. 
The first approach of \citet{huser-genton-16} is based on spatially varying $2 \times 2$ covariance matrices  $\Omega(s), s \in S \subset \mathbb{R}^2$, which are included in the quadratic form 
\begin{equation*}
Q(s_1,s_2)= h^{T}\left( \frac{\Omega(s_{1})+ \Omega(s_{2})}{2}\right)^{-1}h. \end{equation*}
For a valid isotropic correlation model on $\mathbb{R}^d$, e.g., the powered exponential family with unit range
\begin{equation}
R(\| h\| ) = \exp\left\{-\| h\|^{\alpha} \right\},
\end{equation}
a valid non-stationary correlation function on $\mathbb{R}^d$ can be obtained by
\begin{equation}
\rho(s_{1},s_{2}) = \lvert\Omega(s_{1})\rvert ^{\frac{1}{4}} \lvert\Omega(s_{2})\rvert^{\frac{1}{4}} \left\lvert \frac{\Omega(s_{1})+ \Omega(s_{2})}{2} \right\rvert^{-\frac{1}{2}} \cdot R\big(Q(s_1,s_2)^{\frac{1}{2}} \big),
\end{equation}
see \citet{paciorek-shervish-06}. 
 \citet{huser-genton-16} propose the construction of a non-stationary extremal-$t$ process based on Gaussian random fields with such a non-stationary correlation function
 analogously to Section \ref{sec:extremalt}. 
For $s \in \RR^d$, the covariance matrices may be chosen as
\begin{equation*}
\Omega(s) = 
\begin{pmatrix} 
\omega_{x}^{2}(s) & \omega_{x}(s)\omega_{y}(s)\delta(s)\\
 \omega_{x}(s)\omega_{y}(s)\delta(s) & \omega_{y}^{2}(s)\\
\end{pmatrix},
\quad
\end{equation*}
where $\omega_{x}(s) > 0, \omega_{y}(s) > 0$ are the dependence ranges and $\delta(s) \in (-1,1)$ measures the local anisotropy level. 
The non-stationary dependence structure is modelled by including important covariates in the dependence ranges and the anisotropy parameter through link functions. 
For more details, we refer to
\citet{huser-genton-16}.

The second approach of \citet{blanchetdavison2011} apply the stationary
max-stable models of \citet{smith1990max} and \citet{schlather2002models} on an extended and transformed
space, which allows geometric anisotropy. 
More precisely, for a valid isotropic correlation model $R(\| \cdot\|)$ on $\RR^3$, considering max-stable processes based on isotropic correlation functions, they propose to consider a correlation function of the type $R(\|B\cdot\|)$ with 
\begin{align}
    \label{anitransmat}
		&B=\left( \begin{array}{ccc}q_1 \cos(\vartheta) & q_1\sin(\vartheta) & 0 \\ -q_2\sin(\vartheta) 				& q_2\cos(\vartheta)  & 0 \\ 0 & 0 & q_3\end{array} \right),
\end{align}
where $q_1,q_2,q_3>0$ and the components of $S$ correspond to longitude, latitude and altitude. This model is stationary in three-dimensional space, but it is non-stationary in the underlying two-dimensional space.

Besides Smith processes, they apply the resulting non-stationary correlation function also to Schlather processes, i.e., max-stable processes where the stochastic process $W(\cdot)$ in \eqref{spectralrepresentation} has the form
\begin{equation*}
\label{schlather_Wi}
 W(s) = 2\pi^{\frac{1}{2}} \cdot \max\{0,\varepsilon_{i}(s)\}, \quad s \in \mathbb{R}^d. 
\end{equation*}
Note that Schlather processes are special cases of extremal-$t$ processes with $\nu=1$ in \eqref{et_Wi}.

Analogously to the idea of \citet{blanchetdavison2011} we introduce a more general novel non-stationary approach in Chapter \ref{Novel approach for Non-Stationarity}, where we include multivariate covariates in flexible classes of valid variogram and correlation functions in higher dimensions allowing both more individual and mixed effects of the different covariates. Hence, we obtain flexible classes of valid non-stationary models for Brown-Resnick and the extremal-$t$ processes.

\section{Novel approach for Non-Stationarity}
\label{Novel approach for Non-Stationarity}
Henceforth, we use the following notation: For some $d$-dimensional vector $s = (s_i)_{i=1}^d \in \RR^d$ and some index set $I \subset \{1,\ldots,d\}$, we write $s_I = (s_i)_{i \in I}$ to denote the vector build via the components in the set $I$.
In this section, we propose non-stationary models for Brown-Resnick and extremal-$t$ processes. The Brown-Resnick process depends on a variogram, i.e.,  a symmetric function $\gamma: \RR^d \times \RR^d \to [0,\infty)$ that satisfies $\gamma(s,s)=0$ for all $s \in \RR^d$ and is  conditionally negative definite function, that is,$$\label{negdef}
\sum\nolimits_{i=1}^n \sum\nolimits_{j=1}^n a_{i} a_{j} \gamma(s_{i},s_{j}) \leq 0$$ 
for all $a_{1},\ldots,a_{n} \in \RR$  with $\sum\nolimits_{i=1}^n a_{i}=0$, $s_{1},\ldots,s_{n} \in \RR^d$, $n \in \NN$.

In contrast, the extremal-$t$ model depends on a correlation function, i.e., a symmetric function $\rho: \RR^d \times \RR^d \to [-1,1]$ that satisfies $\rho(s,s)=1$  for all $s \in \RR^d$ and is
 positive definite, that is,
$$\sum\nolimits_{i=1}^n \sum\nolimits_{j=1}^n a_{i} a_{j} \rho(s_{i},s_{j}) \geq 0$$ for all $s_{1},\ldots,s_{n} \in \RR^d$ \text{and} $a_{1},\ldots,a_{n} \in \RR$, $n \in \NN$.

Our aim is to construct valid conditionally negative definite functions $\gamma$ and positive definite functions $\rho$ that do not only depend on the separation vector $h$ by including multivariate covariates -- described via a function $c$ -- in valid variogram and correlation functions, respectively.

\begin{prop}
\label{propnewmodel}
Let $I_1,\ldots,I_k \subset \{1,\ldots,p\}$ be arbitrary index sets and let $c: \RR^d \to \RR^p$ be an arbitrary function. Then, for all matrices $A_0 \in \RR^{d \times d }$, $A_j \in \RR^{\lvert I_j \rvert \times \lvert I_j \rvert}$, j=1,\ldots,k, and $\alpha_0, \alpha_1,\ldots,\alpha_k \in (0,2]$ and $\beta \in (0,1]$, the function
\begin{align}
\label{newmodel_propos1}
 \gamma:{}& \RR^d \times \RR^d \to [0,\infty), \nonumber\\
 \gamma(x,y) ={}& \left(\|A_{0} (x - y)\|^{\alpha_0}+\sum\nolimits_{j=1}^{k} \|A_j (c(x)_{I_j} - c(y)_{I_j})\|^{\alpha_{j}}\right)^\beta, 
\end{align}
is a valid variogram. 

Consequently, the function $\rho: \RR^d \times \RR^d \to (0,1]$, $\rho(x,y) = \exp(- \gamma(x,y))$ is a valid correlation function.
\end{prop}

\begin{proof}
Let $\tilde{I}_0=\{1,\ldots,d\}$ and $\tilde{I}_j = \{i + d, \, i \in I_j\}$, $j=1,\ldots,k,$ and $\tilde{c}: \RR^d \to \RR^{d+p}$ with $\tilde{c}(x)=(x,c(x))$. Then, we have
$\gamma(x,y) = \tilde{\gamma}(\tilde c(x),\tilde c(y))$ where
\begin{align}
\label{vario_old}
 \tilde{\gamma}:{}& \RR^{d+p} \times \RR^{d+p} \to [0,\infty), \nonumber\\
 \tilde{\gamma}(x,y) &{}=  \left(\sum\nolimits_{j=0}^{k} \|A_j (x_{\tilde{I}_j} - y_{\tilde{I}_j})\|^{\alpha_{j}}\right)^\beta.
\end{align} 
Thus, it suffices to show that $\tilde \gamma$ is a valid variogram on $\RR^{d+p} \times \RR^{d+p}$.
Firstly, consider the squared Euclidean norm
\begin{align*}
    f:{}& \RR^{d+p} \times \RR^{d+p} \to [0,\infty), \quad 
    f(x,y) ={} \|x-y\|^{2},
\end{align*}
which is conditionally negative definite. 
It is well known that the function 
\begin{align*}
 (x,y) \mapsto f(A_j x_{\tilde{I}_j},A_j y_{\tilde{I}_j})
\end{align*}
is conditionally negative definite on $\RR^{d+p} \times \RR^{d+p}$.
By \cite{berg2007}, we know that the function 
\begin{align*}
    h:{}& \RR^{d+p} \times \RR^{d+p} \to [0,\infty)  \\
     h(x,y)=&{}(f(x,y))^{\alpha_j/2}=
     (\|A_j (x_{\tilde{I}_j} - y_{\tilde{I}_j})\|^2)^{\alpha_j/2}=\|A_j (x_{\tilde{I}_j} - y_{\tilde{I}_j})\|
     ^{\alpha_{j}} ,
\end{align*}
is conditionally negative definite for all $\alpha_{j}/2 \in (0,1]$.
In addition, the property is maintained if these terms are summed up for all $j \in \{1,\ldots,k\}$.
Thus, 
\begin{align*}
 \hat{\gamma}:{}& \RR^{d+p} \times \RR^{d+p} \to [0,\infty), \\
 \hat{\gamma}(x,y) =&{} \left(\sum\nolimits_{j=1}^k \|A_j (x_{\tilde{I}_j} - y_{\tilde{I}_j})\|^{\alpha_j}\right), 
\end{align*}
is conditionally negative definite.
Repeating the above argument of \cite{berg2007}, the same holds true for $\tilde{\gamma}(x,y) = \hat{\gamma}(x,y)^{\beta}$ with $\beta \in (0,1]$ in \eqref{vario_old}.
\end{proof} 

\begin{ex}
Applying the construction of the paper by  \citet{blanchetdavison2011}
directly to the power-type variogram would have resulted in a variogram model of the type
\begin{equation*}
\label{ex_propos1}
 \gamma(x,y) = \left(\|A (x - y)\|^{2}+\|q_3 (c(x)- c(y))\|^{2}\right)^\beta, 
\end{equation*}
where $A$ is of the same form as in \eqref{animat} and $c: \RR^2 \to \RR$. 
This is a special case of Proposition \ref{propnewmodel}, if $\alpha_0=\alpha_1=2$, $k=1$ and $I_1=\{1\}$. However, Proposition \ref{propnewmodel} is more general and allows multiple covariates to be considered simultaneously with various ways to interact among each other and with the coordinates.
\end{ex}
For non-degenerate $A_0$, the variogram in Proposition \ref{propnewmodel} is unbounded with $\lim_{\|s_1-s_2\| \to \infty}\gamma(s_1,s_2)=\infty.$
Then, for the extremal coefficient of the Brown-Resnick process associated to this variogram, it holds that $\lim_{\|s_1-s_2\| \to \infty} \theta(s_1,s_2)=2$, or equivalently $\lim_{\|s_1-s_2\| \to \infty}\chi(s_1,s_2)$=0. Thus, this model assumes that the asymptotic dependence can get arbitrarily weak at large distances. However, a bounded variogram might be more reasonable in certain applications. Proposition $\ref{prop_schlather}$ introduces a valid variogram model based on \citet{schlather2017}. The variogram model results in a bounded variogram for $\alpha < 0$ and an unbounded variogram for $\alpha \geq 0$.
Moreover, it is a generalization of Proposition $\ref{propnewmodel}$ if $\alpha=\beta$. 

\begin{prop}
\label{prop_schlather}
Let $I_1,\ldots,I_k \subset \{1,\ldots,p\}$ be arbitrary index sets and $c: \RR^d \to \RR^p$ be an arbitrary function. Then, for all matrices $A_0 \in \RR^{d \times d }$, $A_j \in \RR^{\lvert I_j \rvert \times \lvert I_j \rvert}$, $j=1,\ldots,k$, and $\alpha_0, \alpha_1, \ldots, \alpha_k \in (0,2]$, $\beta \in (0,1]$ and $\alpha \in (-\infty,1]$, the function
\begin{align}
\label{newmodel2}
 \gamma:& \RR^d \times \RR^d \to [0,\infty), \nonumber\\
     \gamma(x,y) =&{}\frac{\left(1+  \left(\|A_{0} (x - y)\|^{\alpha_0}+\sum\nolimits_{j=1}^{k} \|A_j (c(x)_{I_j} - c(y)_{I_j})\|^{\alpha_{j}}\right)^{\beta}\right)^{\alpha/\beta}-1}{2^{\alpha/\beta}-1}, 
\end{align}
is a valid variogram. Consequently, the function $\rho: \RR^d \times \RR^d \to (0,1]$, $\rho(x,y) = \exp(- \gamma(x,y))$ is a valid correlation function.
\end{prop}
\begin{proof}
The proof is an adaption of the proof of Proposition 1 in \citet{schlather2017}.
By \citet{schlather2017}, the function $f(g)=(1+g^{\beta})^{\frac{\alpha}{\beta}}$
is a complete Bernstein function for any $0< \beta \leq 1$, $0 < \alpha \leq 1$ and for any conditionally negative definite function $g \geq 0$. By Proposition \ref{propnewmodel},
the function $g(x,y)=\|A_{0} (x - y)\|^{\alpha_0}+\sum\nolimits_{j=1}^{k} \|A_j (c(x)_{I_j} - c(y)_{I_j})\|^{\alpha_{j}}$ is such a conditionally negative definite function. Consequently, \eqref{newmodel2} is a valid variogram by the same arguments of \cite{schlather2017}, that is, one uses that constants are conditionally negative definite functions and the set of conditionally negative definite functions forms a cone. 

In the case that $\alpha=0$, the limiting function is a variogram by the characteristics of conditionally negative functions. 

For $\alpha<0$ and a variogram $g$, the function $(1+g)^{\frac{\alpha}{\beta}}$ is positive definite and thus $(1-(1+g)^{\frac{\alpha}{\beta}})/(1-2^{\frac{\alpha}{\beta}})$ is a variogram, see \cite{schlather2017}. 

\end{proof}

\section{Application to heavy rainfall data}
\label{Application}

\subsection{Data description}
We use a data set of historical daily precipitation observations for two regions
delivered by Germany's National Meteorological Service, the Deutscher Wetterdienst (DWD). In our case study, we consider observations from 72 weather stations for the region in Southern Germany and 46 stations for the region in Northern Germany, see Figure \ref{weatherstations}. 
For each station, we look at the daily precipitation height (measured in mm) of the the summer months (June, July and August) over 69 years, that is from 1951 to 2019, and divide these data into yearly blocks. This then gives us 69 annual maxima of daily summer precipitation for each station, which we then use to model the distribution of the block maxima via max-stable processes. In addition to longitude and latitude, which we transform such that the Euclidean distances correspond to the distances in km, we use the altitude (for numerical reasons also in km) of the stations as covariate in the non-stationary models.
As can be seen in the topographical map in Figure \ref{weatherstations}, the region in Northern Germany is flatter than the region in Southern Germany.  

\begin{figure}[h!]
\centering
\includegraphics[width=0.75\textwidth]{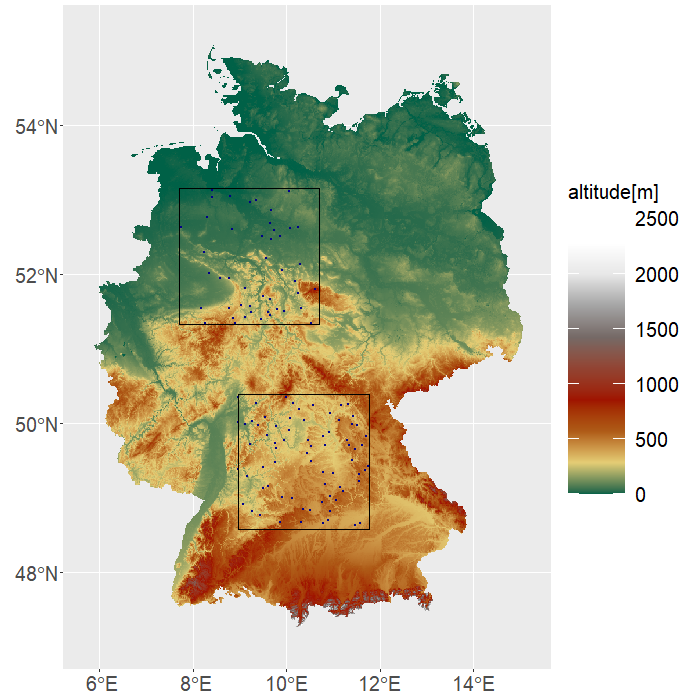}
\caption{Topographical map of Germany with the 118 weather stations (in blue) in Germany considered in our case study. The two study regions are marked by the black boxes, containing 46 stations for the region in Northern Germany and 72 stations for the region in Southern Germany.}
\label{weatherstations}
\end{figure}

\subsection{Model comparison}
\label{modelfittingandcomparison}

In order to narrow the range of relevant spatial models, we perform a preliminary analysis of the spatial dependence structure. In Figure \ref{excof}, it can be seen that the empirically estimated pairwise extremal coefficients 
based on the method of \citet{caperaa1997nonparametric} and implemented in the R package evd \citep[see][]{evd_citation} 
are usually well below two. This justifies the use of asymptotically dependent models such as max-stable models. For small distances, the extremal coefficients do not approach the value 1, therefore we incorporate some nugget effect into our model. Consequently, we utilize two of the most popular max-stable models, which are Brown-Resnick and extremal-$t$ and investigate several stationary and non-stationary variants with a nugget parameter. 

  \begin{figure}[htp]
    \centering
    \begin{subfigure}{0.49\textwidth}
      \centering
      \includegraphics[width=0.99\textwidth]{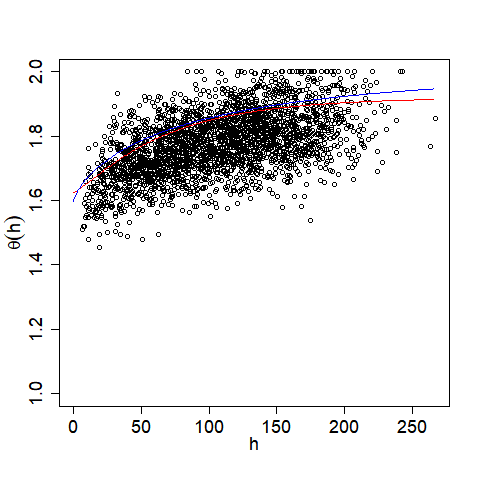}
      \caption{Region of Southern Germany}
    \end{subfigure}%
    \hfill
    \begin{subfigure}{0.49\textwidth}
      \centering
      \includegraphics[width=0.99\textwidth]{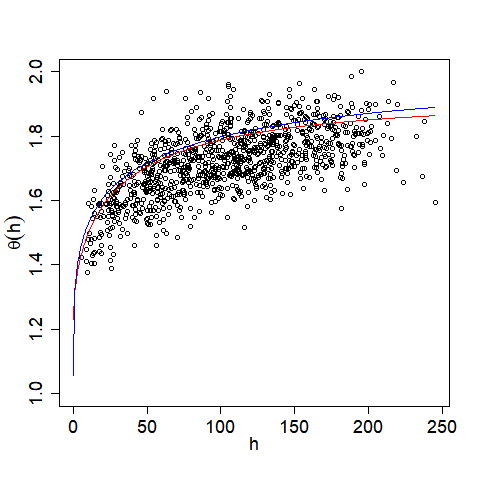}
      \caption{Region of Northern Germany}
    \end{subfigure}
    \caption{Comparison between the fitted (theoretical) extremal coefficient function of the extremal-$t$ model (red line), the fitted (theoretical) extremal coefficient function of the Brown-Resnick model (blue line) and the estimated extremal coefficients based on the F-madogram (black points) for the two isotropic models.}
    \label{excof}
\end{figure}
 
More specifically, we consider Brown-Resnick processes associated to variograms of the type
\begin{align}
\label{vario_final}
 \gamma:{}& \RR^d \times \RR^d \to [0,\infty), \nonumber\\
 \gamma(x,y) ={}& \sigma_{nug}^2\mathds{1}\{\|x-y\| \neq 0\} +  \tilde{\gamma}(x,y)
\end{align}
where $\sigma_{nug}^2 \in [0,\infty)$ is the size of the nugget effect and  $\tilde{\gamma}: \RR^d \times \RR^d \to [0,\infty)$ is a continuous variogram. The models we consider are based on different choices for $\tilde \gamma$.
In addition, whenever $\tilde \gamma: \RR^d \times \RR^d \to [0, \infty)$ is a valid variogram,
\begin{align}
\label{validcorr}
 \tilde{\rho}: \ & \RR^d \times \RR^d \to (0,1], \ \tilde{\rho}(x,y) = \exp(- \tilde{\gamma}(x,y)), 
\end{align}
is a valid correlation model. Thus, the extremal-$t$ models we will consider are based on correlation functions of the type
\begin{align}
\label{corr_final}
 \rho:{} \RR^d \times \RR^d \to [0,1], \
 \rho(x,y) ={} \sigma_{nug}^2\mathds{1}\{\|x-y\| = 0\} +  (1-\sigma_{nug}^2)\tilde{\rho}(x,y),
\end{align}
where $\sigma_{nug}^2 \in [0,1]$ and we insert the previously mentioned conditionally negative definite functions $\tilde{\gamma}$ into \eqref{validcorr} to obtain analogues to the Brown-Resnick models mentioned above.

\subsubsection{Stationary models}

For the isotropic extremal-$t$ model, we utilize for $\tilde{\rho}$
the correlation function \eqref{iso_corr} and $\tilde{\gamma}$ is given by the variogram \eqref{iso_var} for the isotropic Brown-Resnick model. In addition, we take into consideration the geometric anisotropic model with correlation function  \eqref{ani_corr}.
We obtain valid
variogram models from the correlations \eqref{validcorr} as follows:
\begin{align*}
 \tilde \gamma: \RR^d \times \RR^d \to [0, \infty), \ \tilde{\gamma}(x,y) = -\log(\tilde{\rho}(x,y)). 
\end{align*} 

\subsubsection{Non-stationary models}
The model of Proposition \ref{propnewmodel} can be applied to the case where the components $(x_1,x_2)$ correspond to geographical coordinates,
while $c_1,\ldots,c_k$ are covariates. Thus, it is a natural choice to partition the components via the subsets $I_1=\{1\}$, $I_2=\{2\}, \ldots, I_{k} = \{k\}$. Then, $A_0$ can be chosen as the general 
anisotropy matrix $A$ from \eqref{animat} 
while $A_1,\ldots,A_{k}$ are just non-negative real values.
The model is identifiable since
the sets $I_{1},\ldots,I_{k}$ are pairwise disjoint. 
We consider several non-stationary 
variants of that novel non-stationary approach for the extremal-$t$ and Brown-Resnick model
based on only single covariate $c_1$ denoting the altitude similarly to \citet{blanchetdavison2011} and \citet{huser-genton-16}. 
In the nonstationary model $M_1$, the function $\tilde{\gamma}$ is given by 
\begin{align}
 \tilde{\gamma}(x,y) =&{}  \|A_{0} (x-y)\|^{\alpha_0}+\|q_{3} (c(x)-c(y))\|^{\alpha_0},\tag{$M_1$}
\end{align}
with $\alpha_0 \in (0,2]$ and $q_{3} \geq 0$. 

A more general model is the nonstationary model $M_2$ with
\begin{align}
 \tilde{\gamma}(x,y) ={}&  \left(\|A_{0} (x-y)\|^{\alpha_0}+\|q_{3} (c(x)-c(y))\|^{\alpha_1}\right)^\beta,\tag{$M_2$}
\end{align}
where $\alpha_0,\alpha_1 \in (0,2]$, $q_{3} \geq 0$ and $\beta \in (0,1]$.

The previous mentioned models are valid variogram models by Proposition \ref{propnewmodel}. 
The generalization to the variogram model in Proposition \ref{prop_schlather} with $d=2$ is called nonstationary model $M_3$, that is,
\begin{align}
     \tilde{\gamma}(x,y) ={}& \frac{\Big(1+  \left(\|A_{0} (x-y)\|^{\alpha_0}+\|q_{3} (c(x)-c(y))\|^{\alpha_1}\right)^{\beta}\Big)^{\alpha/\beta}-1}{2^{\alpha/\beta}-1}, \tag{$M_3$}
\end{align}
where $\alpha_0,\alpha_1 \in (0,2]$, $q_{3} \geq 0$, $\beta \in (0,1]$ and $\alpha \in (-\infty,1]$.

Moreover, we fit the nonstationary model $M_2$ with fixed $\alpha_0=\alpha_1=2$, which is of the same type as the model in \cite{blanchetdavison2011} for the extremal-$t$ and Brown-Resnick model and denote it by $M_{BD}$.
In addition, we consider the non-stationary extremal-$t$ model by \cite{huser-genton-16} with altitude as covariate in the dependence ranges and denote this model by $M_{HG}$.

\subsection{Inference} \label{subsec:inference}

Denote by $n$ the number of years and let $k$ be the number of locations. Additionally, the annual maximum at location $s_{i}$ for year $m$ is described by $z_{m}(s_{i})$ and the contribution of these data to the bivariate density is $f_{s_{i},s_{j}}(z_{m}(s_{i}),z_{m}(s_{j}))$. 
The dependence parameters are then fitted by maximizing the pairwise log-likelihood function:
\begin{equation*}
	\ell_{p}(z;\psi) = \sum_{m=1}^n  \sum_{i,j=1,i \neq j}^{k} \log(f_{s_{i},s_{j}}(z_{m}(s_{i}),z_{m}(s_{j}))).
\end{equation*} 

We try several optimization methods and report the best results for each model. More precisely, we use an iterative optimization procedure such that the pairwise log-likelihood improves gradually based on previous model parameters.
This means that for the anisotropic model, the model $M_{HG}$, and the non-stationary models $M_2$ and $M_3$, each of which is a generalization of some other model, we use the optimal parameters of the corresponding model, which for fixed parameters corresponds to the more general model. For the non-stationary model $M_1$, we 
optimize in several steps, i.e. we start with the optimization of a part of the parameters and set the others to fixed values. Then we include more parameters in the optimization procedure until we finally optimize over all parameters.

\subsection{Results of the model comparison}
In order to check if the isotropic models describe adequately the extremal dependence for bivariate extreme events, we compare the empirically estimated pairwise extremal coefficients with the theoretical coefficients according to the fitted extremal-$t$ and Brown-Resnick models in Figure \ref{excof}. 
It can be seen that the fit of the coefficients seems to be quite good for the region of Northern Germany. Although the fit of the region of Southern Germany seems to be not as good as for the region of Northern Germany, is is also reasonable. In addition to the isotropic model, we consider the other models mentioned in Section \ref{modelfittingandcomparison} and compare all of them with an information criterion,
so called Takeuchi’s information criterion (TIC)  by \citet{takeuchi1976distribution}, which evaluates the model fit and additionally takes into account the model 
complexity. Then, the model with the lowest TIC value is selected. We consider the previously mentioned stationary and non-stationary extremal-$t$ (ET) and Brown-Resnick (BR) processes for the two regions. As shown in Table \ref{tab:table1}, for the region in Southern Germany, all non-stationary models outperform the stationary models and the non-stationary model $M_1$ results in the best TIC value for the two processes.
In contrast, for the region in Northern Germany, Table \ref{tab:table1_part2} indicates that the isotropic 
model is the best one for both processes. 

The difference between the best extremal-$t$ models in both regions also becomes evident when looking at exact realizations of the models generated by the extremal functions algorithm \citep{dombry2016exact} as displayed in Figure \ref{Sim}. For the region Southern Germany, we see a higher variability of the process at smaller scales compared to Northern Germany. This is in line with the topographical features of these two regions, see Figure \ref{weatherstations}.
\begin{figure}[h!]
\centering
\begin{subfigure}{0.49\textwidth}
      \centering
\includegraphics[height=5cm]{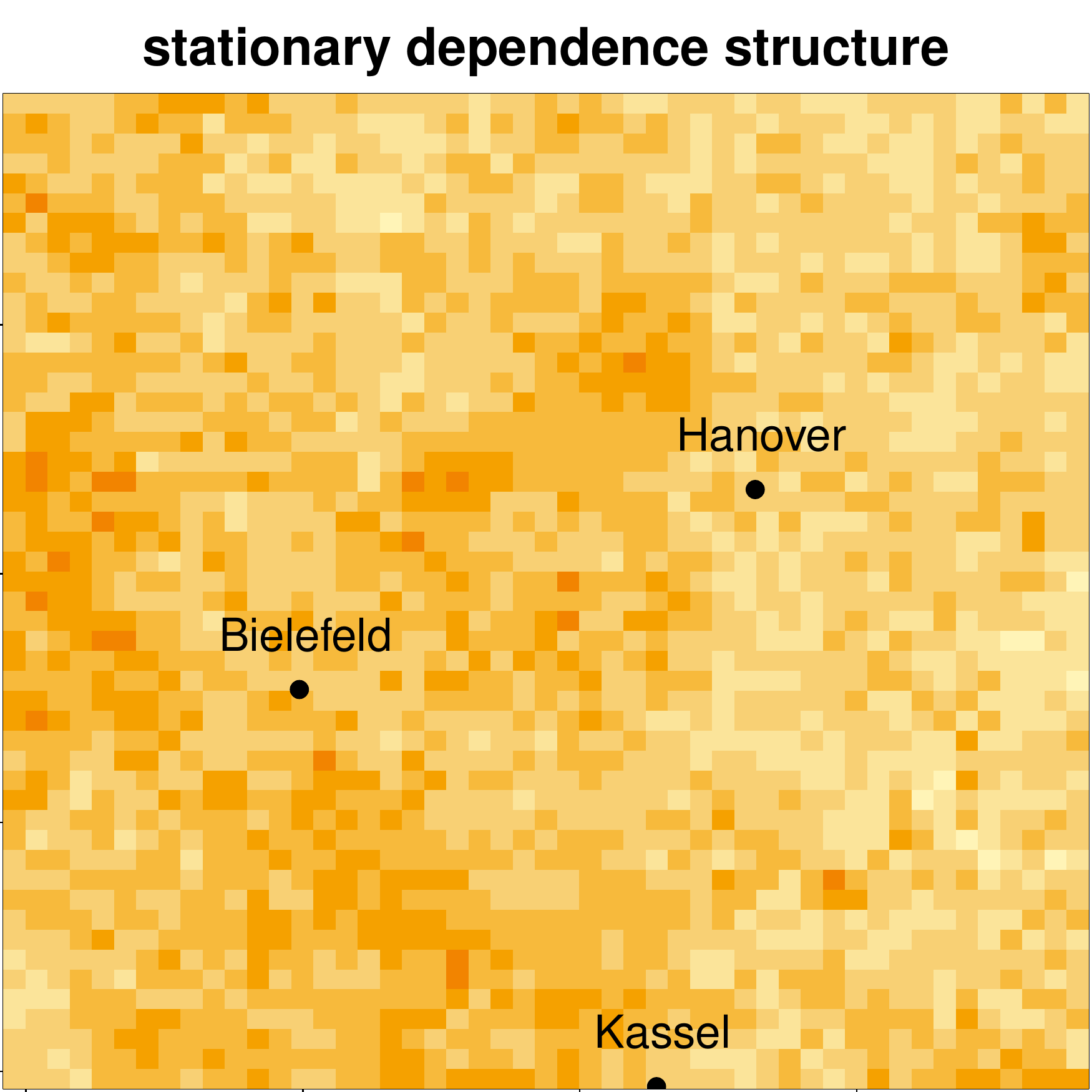}
      \hspace{.15cm}
       \caption{Region of Northern Germany}
        \end{subfigure}%
    \hfill
    \begin{subfigure}{0.49\textwidth}
      \centering
       \includegraphics[height=5cm]{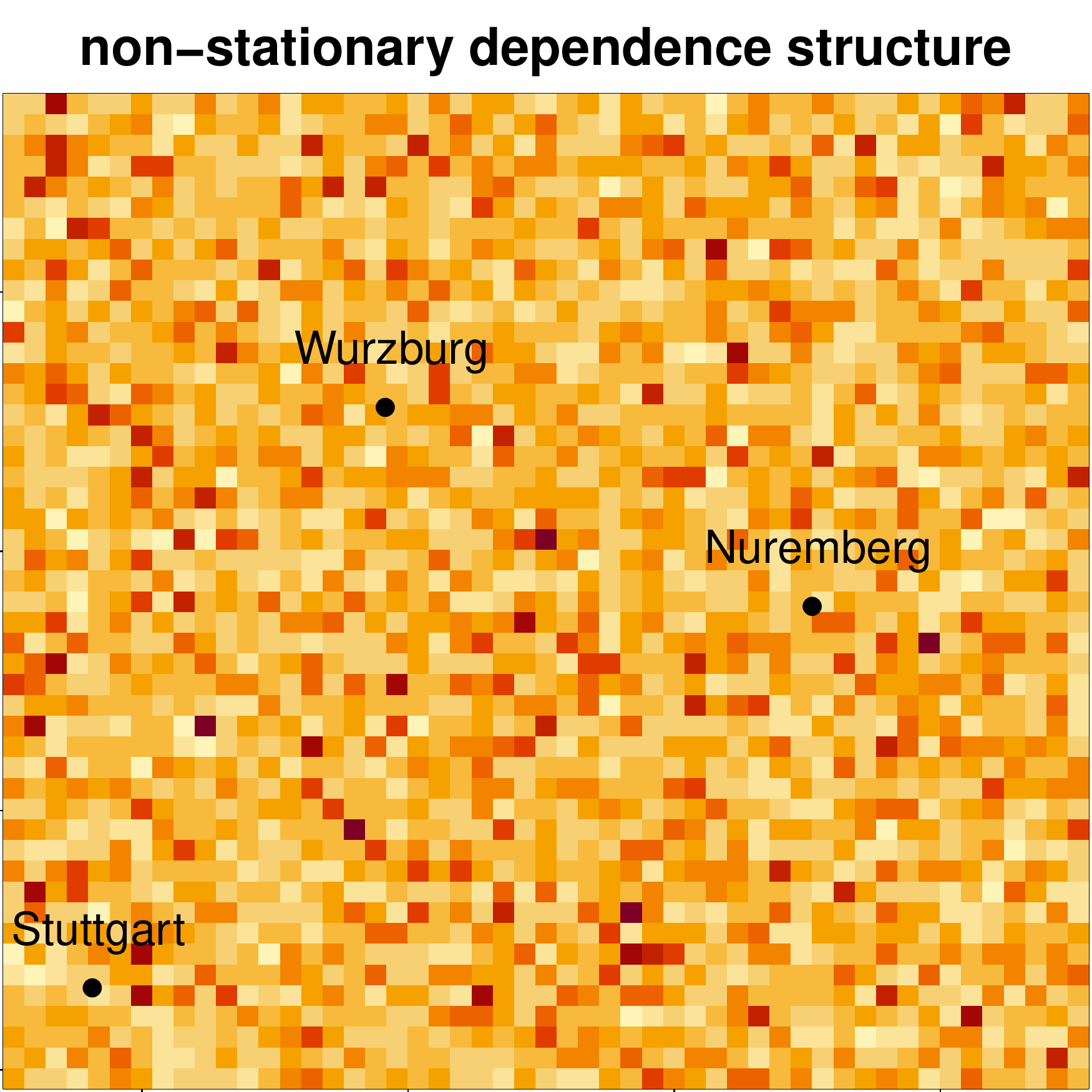}
        \caption{Region of Southern Germany}
            \end{subfigure}
\caption{Simulated realizations of the best extremal-$t$ model in both regions considered.}
\label{Sim}
\end{figure}   

Moreover, the comparison between extremal-$t$ and Brown-Resnick models shows that the extremal-$t$ models 
work better for both regions. 
The generalization to the non-stationary $M_3$ does not improve the model fit in our considered case study because we have an unbounded variogram. However, it might be useful to have more flexibility and also allow a bounded variogram.

\begin{table}[h!]
  \begin{center}
    \caption{TIC and optimal loglikelihood values for Southern Germany.}
    \label{tab:table1}
    {\small
     \begin{tabular}{|l|l|l|l|c|} 
      \hline
      \textbf{Model} & \textbf{TIC of ET} & \textbf{TIC of BR} & \textbf{LogL of ET} & \textbf{LogL of BR}\\
      \hline
      isotropic model & 1\,514\,056\,.0  & 1\,514\,452\,.0 & -756\,876\,.4 & -757\,087\,.8\\
      anisotropic model &  1\,514\,001\,.0 &  1\,514\,415\,.0 &  -756\,816\,.8 & -757\,035\,.9\\ 
      non-stat.~model $M_1$ & \textbf{1\,513\,963\,.0} &  \textbf{1\,514\,378\,.0} & -756\,797\,.7 & -757\,015\,.0  \\
      non-stat.~model $M_{2}$ &   1\,513\,979\,.0  &  1\,514\,396\,.0 &   -756\,797\,.7 & -757\,015\,.0\\
      non-stat.~model  $M_{3}$ &  1\,513\,994\,.0  & 1\,514\,404\,.0 &   -756\,797\,.7 &  -757\,015\,.0\\
      non-stat.~model $M_{BD}$ &  1\,513\,973\,.0 & 1\,514\,380\,.0 & -756\,800\,.4 &-757\,015\,.2 \\ 
      non-stat.~model $M_{HG}$ &  1\,514\,045\,.5 & - & -756\,808\,.7  & -\\ 
      \hline
    \end{tabular}
}
  \end{center}
\end{table}

\begin{table}[h!]
  \begin{center}
    \caption{TIC and optimal loglikelihood values for Northern Germany.}
    \label{tab:table1_part2}
    {\small
        \begin{tabular}{|l|l|l|l|c|}
      \hline
      \textbf{Model} & \textbf{TIC of ET} & \textbf{TIC of BR} & \textbf{LogL of ET} & \textbf{LogL of BR}\\
      \hline
      isotropic model &  \textbf{608\,495\,.3} & \textbf{608\,768\,.7}  & -304\,129\,.5 & -304\,285\,.8\\
      anisotropic model & 608\,533\,.5  &  608\,810\,.3 & -304\,128\,.4 & -304\,284\,.5\\ 
      non-stat.~model $M_1$ & 608\,544\,.3 & 608\,823\,.2 & -304\,128\,.4 & -304\,284\,.5\\
      non-stat.~model $M_{2}$ &  608\,560\,.9 &  608\,838\,.0   &   -304\,128\,.4 &  -304\,284\,.5\\
      non-stat.~model  $M_{3}$ &  608\,568\,.7 &  608\,847\,.8  &  -304\,128\,.4 &  -304\,284\,.5\\
      non-stat.~model $M_{BD}$ & 608\,538\,.8 & 608\,816\,.4  &  -304\,128\,.4 & -304\,284\,.5\\ 
      non-stat.~model $M_{HG}$ & 608\,538\,.9 & - &  -304\,120\,.7 & -\\ 
      \hline
    \end{tabular}
  }
  \end{center}
\end{table}

Moreover, we verify our procedure and assess its uncertainty via parametric bootstrap. More precisely, we draw 100 samples from the best fitted extremal-$t$ model for the region in Southern Germany, that is the non-stationary model $M_1$, 
and re-estimate its parameters via the procedure described in Section \ref{subsec:inference}. The results are displayed in Table \ref{tab:table3}. We can see that, for each parameter, the absolute difference of the true value and the mean is (much) smaller than the standard deviation which verifies the validity of our procedure. However, it becomes also evident that the parameter $q_3$ is quite difficult to estimate as the high deviation indicates. The uncertainty in the parameter estimation for $\alpha_0$ is probably also linked to this issue.

\begin{table}[h!]
  \begin{center}
    \caption{Parametric bootstrap from the best fitted extremal-$t$ model}
    \label{tab:table3}
    {\small 
    \begin{tabular}{|l|c|c|r|} 
     \hline
      \textbf{parameter} & \textbf{true value} &\textbf{mean (simulations)} & \textbf{sd (simulations)}\\
      \hline
      $\nu$ & 4.094 & 5.539 & 2.726 \\
      $\theta$ & -0.726 & -0.241 & 0.556 \\
      $q_1$ & 0.011 & 0.008 & 0.005 \\
			$q_2$ & 0.006  & 0.006 & 0.003 \\
      $q_3$ & 1.302 & 1.043 & 0.882 \\
      $\alpha_0$ & 1.323 & 1.226 & 0.564\\
	  $\sigma_{nug}^2$ & 0.315 & 0.236 & 0.150 \\
	\hline
    \end{tabular}
  }
  \end{center}
  
\end{table}   
       
\section{Bridging between asymptotic dependence and independence}
\label{bridgebetweenasympandsy}

In Section \ref{Novel approach for Non-Stationarity}, we consider max-stable models with fixed covariates in the dependence structure. By definition, these models are asymptotically dependent. Instead of including fixed covariates, the covariates might also be random processes. These can have consequences on both the margins and the dependence structure. 
To illustrate these effects, we assume that the covariate process $Y = \{Y(s), \ s \in \RR^d\}$ possesses $\alpha$-Pareto distributed marginal distributions and include the covariate process in the dependence structure and in the margins in the following way: 
Conditionally on $Y$, let $Z$ be a max-stable process with marginal distribution
\begin{equation}
\label{z_marg_distr}
\PP(Z(s) \leq z \mid Y) = \PP(Z(s) \leq z \mid Y(s)) = \exp(- Y(s) / z), 
\end{equation}
where $z > 0, \ s \in \RR^d$ and with a dependence structure which may depend on $Y$. Thus, $Y(s)$ is chosen as scale parameter in the marginal distribution for $s \in \RR^d$.
By assumption, conditionally on $Y(\cdot)= y(\cdot)$, we can write $Z(\cdot) = y(\cdot) \cdot Z_y(\cdot)$
   where $Z_y(\cdot)$ is a max-stable process with unit Fr\'echet margins and the same dependence structure as $Z \mid Y = y$.
Thus, unconditionally, we can write 
\begin{equation} \label{eq:randomscale}
Z(\cdot) = Y(\cdot) \cdot Z_Y(\cdot).
\end{equation}

In this section, we investigate the extremal dependence behaviour of vectors $(Y(s_1)Z_Y(s_1),Y(s_2)Z_Y(s_2))$. 

A similar construction is considered in \citet{Engelke2019}. They propose the construction $(RW_1,RW_2)=R(W_1,W_2)$ with a non-degenerate random variable $R >0$ that is independent of the bivariate random vector $(W_1,W_1)$ and present several results on the extremal dependence subject to the tail behaviour of $R$. However, our construction differs from \citet{Engelke2019} as we consider a max-stable process $ Z_Y(\cdot)$ instead of a random variable $R$. In addition, this max-stable process is not necessarily independent from $(Y(s_1),Y(s_2))$. 

Depending on the tails of the covariate processes, we derive different results. Our results indicate that for $\alpha \neq 1$, the extremal dependence of the above mentioned random scale construction does not depend on the dependence structure of $Z_Y$ , but on the actual value of $\alpha$ only. For $\alpha \in (0,1)$, we get asymptotic independence in Theorem \ref{thm 51}, while we obtain asymptotic dependence for $\alpha >1$ in Theorem \ref{thm 52}. In case of $\alpha =1$, asymptotic independence might also occur if the dependence in the max-stable process $Z_Y$ weakens sufficiently fast as $Y$ gets large (Theorem \ref{thm 53}). In case of perfect dependence in $Z_Y$, however, i.e., if $Z_Y$ corresponds to a spatially 
constant Fr\'echet distributed variable $R$, we have asymptotic dependence according to Proposition 6, case 1 of \citet{Engelke2019}.

Note that the (conditionally) max-stable processes considered in this section, will not necessarily be sample-continuous. Thus, henceforth, we will turn down any assumptions on the sample paths and consider the general set $\mathcal{F}$ of all functions $f:\RR^d \to \RR$, equipped with  the $\sigma$-algebra generated by the cylinder sets \eqref{cylinder_sets}. 
Consequently, max-stability in $\mathcal{F}$ coincides with max-stability w.r.t.\ all finite-dimensional distributions.

\begin{thm}
\label{thm 51}
Let $Y = \{Y(s), \ s \in \RR^d\}$ be a covariate process whose components are asymptotically independent and $\alpha$-Pareto distributed with some $\alpha \in (0,1)$. Furthermore, let $Z$ be of the form \eqref{eq:randomscale} where, conditionally on $Y$, $Z_Y$ is a max-stable process with unit Fr\'echet margins and a dependence structure that may depend on $Y$.
Then, the process $Z$ is asymptotically independent, i.e.,
$$ \chi_Z(s_1,s_2) = \lim_{u \to \infty} \PP(Z(s_2) > u \mid Z(s_1) > u) = 0 \quad \forall s_1, s_2 \in \RR^d. $$
\end{thm}
\begin{proof}

Denote by $ \tilde{Z} $ a unit Fr\'echet distributed random variable that is independent from $Y$.
Consequently, it holds $\bar{F}_{\tilde{Z}} \in RV_{-1}^{\infty}$ while $\bar{F}_{Y} \in RV_{-\alpha_{Y}}^{\infty}$ with $1=\alpha_{\tilde{Z}} > \alpha_{Y}$.

It holds 
\begin{align*}
&\PP(Z(s_1) \leq u,Z(s_2) \leq u) {}=\int_{\mathcal{F}} \PP(y(s_1) Z_{y}(s_1) \leq u, y(s_2) Z_{y}(s_2) \leq u)\,\PP_{Y}(\mathrm{d}y)  \\
&{}=\int_{\mathcal{F}}\PP\left(Z_{y}(s_1)\leq \frac{u}{y(s_1)},Z_{y}(s_2) \leq \frac{u}{y(s_2)}\right)\,\PP_{Y}(\mathrm{d}y)\\
&\leq  \int_{\mathcal{F}}\PP\left(\tilde{Z} \leq \min\left\{\frac{u}{y(s_1)},\frac{u}{y(s_2)}\right\} \right)\,\PP_{Y}(\mathrm{d}y)
{}=\PP(\tilde{Z}Y(s_1)\leq u, \tilde{Z} Y(s_2)\leq u).
\end{align*}

Moreover, for all $s \in \RR^d$, it is
\begin{equation*}
\PP(Z(s) \leq u) = \PP(Z_{Y}(s)Y(s) \leq u) = \PP(\tilde{Z}Y(s) \leq u)=\int_{\mathcal{F}}\exp\Big( -\frac{y(s)}{u}\Big)\PP_{Y}(\mathrm{d}y).
\end{equation*}

Consequently, it is
\begin{align*}
&\PP(Z(s_1)>u,Z(s_2)>u)= \PP(Z_{Y}(s_1)Y(s_1)>u,Z_{Y}(s_2)Y(s_2)>u)\\
&{}=1 - \PP(Z_{Y}(s_1)Y(s_1) \leq u) - \PP(Z_{Y}(s2)Y(s_2) \leq u) \\ 
& \quad +\PP(Z_{Y}(s_1)Y(s_1) \leq u,Z_{Y}(s_2)Y(s_2)\leq u)\\
&{}\leq 1 - \PP(\tilde{Z}Y(s_1) \leq u) - \PP(\tilde{Z}Y(s_2) \leq u) +\PP(\tilde{Z}Y(s_1)\leq u,\tilde{Z}Y(s_2)\leq u)\\
&{}= \PP(\tilde{Z}Y(s_1)>u, \tilde{Z} Y(s_2)>u).
\end{align*}

By \citet[][Proposition 5]{Engelke2019}, it holds
\begin{equation*}
\lim_{u \to \infty} \PP(\tilde{Z}Y(s_1)>u \mid \tilde{Z} Y(s_2)>u) =  \lim_{u \to \infty} \PP(Y(s_1)>u \mid Y(s_2)>u)=0.
\end{equation*}
Thus, it holds  
\begin{align*}
 \chi_{Z}(s_1,s_2) ={}& \lim_{u \to \infty} \PP(Z(s_1)>u \mid Z(s_2)>u)\\
\leq{}& \lim_{u \to \infty} \frac{\PP(\tilde{Z}Y(s_1)>u,\tilde{Z}Y(s_2)>u)}{\PP(\tilde{Z}Y(s_2)>u)} = 0 .
\end{align*}
\end{proof}

\begin{thm}
\label{thm 52}
Let $Y = \{Y(s), \ s \in \RR^d\}$ be a covariate process whose components are $\alpha$-Pareto distributed
with $\alpha>1$.
Furthermore, let $Z$ be of the form \eqref{eq:randomscale} where, conditionally on $Y$, $Z_Y$ is a max-stable process with unit Fr\'echet margins and a dependence structure that may depend on $Y$.
In addition, let the tail dependence coefficient of $Z_y$ be strictly positive  for all $y \in \mathcal{F}$. 
Then, the process $Z$ is asymptotically dependent, i.e.,
$$ \chi_{Z}(s_1,s_2) = \lim_{u \to \infty} \PP(Z(s_2) > u \mid Z(s_1) > u) >0
\quad \forall s_1, s_2 \in \RR^d. $$
\end{thm}
\begin{proof}

By Breiman's Lemma \citep{breiman1965}, it holds 
\begin{equation*}
\PP(Z_{Y}(s)Y(s) > u)
\sim \EE Y(s) \PP(Z_{Y}(s) > u), \quad u \to \infty,
\end{equation*}
which implies that  $\bar{F}_{Z_{1}} \in RV_{-1}^{\infty}$ for all $s \in \RR^d$.

In addition, using that $Y(s_1), Y(s_2) \geq 1$ a.s., we obtain
\begin{align*}
&\PP(Z_{Y}(s_2)Y(s_2) > u,Z_{Y}(s_1)Y(s_1) > u)\\
&{}= \PP(\min\{Z_{Y}(s_2)Y(s_2),Z_{Y}(s_1)Y(s_1)\} > u  )\\
&{}\geq \PP(\min\{Z_{Y}(s_2),Z_{Y}(s_1)\} > u) {}=\PP(Z_{Y}(s_2) > u,Z_{Y}(s_1) > u).
\end{align*}

Thus, it holds
\begin{align*}
&\lim_{u \to \infty} \frac{\PP(Z_{Y}(s_2)Y(s_2) > u,Z_{Y}(s_1)Y(s_1) > u)}{\PP(Z_{Y}(s_1)Y(s_1) > u)}\\
&{}\geq \lim_{u \to \infty} \frac{\PP(Z_{Y}(s_2) > u,Z_{Y}(s_1) > u)}{ \EE Y(s_1) \cdot \PP(Z_{Y}(s_1)> u)}
{}={} \lim_{u \to \infty} \frac{\PP(Z_{Y}(s_2) > u,Z_{Y}(s_1) > u)}{ \EE Y(s_1) \cdot (1-\exp(-1/u))} \\ 
&{}= \frac{1}{\EE Y(s_1)} \int_{\mathcal{F}} \lim_{u \to \infty} \PP(Z_{y}(s_2) > u \mid Z_{y}(s_1) > u) \,\PP_{Y}(\mathrm{d}y)\\
&{}= \frac{1}{\EE Y(s_1)} \int_{\mathcal{F}} \chi_{Z_{y}}(s_1,s_2) \,\PP_{Y}(\mathrm{d}y) > 0 
\end{align*}
where we use the dominated convergence with majorant 1 and the fact that $\chi_{Z_{y}}(s_1,s_2)$, the tail dependence coefficient of the random variables $Z_{y}(s_1)$ and $Z_{y}(s_2)$, is strictly positive for all $y \in \mathcal{F}$.

\end{proof}

\begin{thm}
\label{thm 53}
Let $Y = \{Y(s), \ s \in \RR^d\}$ be a covariate process, whereby $Y$ is standard Pareto noise, i.e., we have the joint probability density function
$$ f_{(Y(s_1),\ldots, Y(s_n))} (y_1, \ldots,y_n) = \prod_{i=1}^n y_i^{-2}, \quad y_1,\ldots,y_n > 1, $$
for all pairwise distinct $y_1,\ldots,y_n \in \RR^d$. 
Furthermore, let $Z$ be of the form \eqref{eq:randomscale} where,
conditionally on $Y$, $Z_Y$ is a Brown-Resnick process with unit Fr\'echet margins 
associated to some variogram 
 $\gamma_Y$ satisfying 
 $$\gamma_Y(s_1,s_2) = \gamma(s_1,s_2; Y(s_1),Y(s_2)) > c \| Y(s_1) - Y(s_2)\|^\kappa$$ for all $s_1,s_2 \in \RR^d$ and some $c>0$, $\kappa >1$.

Then, the process $Z$ is asymptotically independent, i.e.,
$$ \chi(s_1,s_2) = \lim_{u \to \infty} \PP(Z(s_2) > u \mid Z(s_1) > u) = 0 \quad \forall s_1, s_2 \in \RR^d. $$
\end{thm}

\begin{proof}
It holds
\begin{align*}
&\PP(Z(s_{1}) \leq u,Z(s_{2}) \leq u)\\
&{}=\int_{\mathcal{F}} \PP(Z(s_{1}) \leq u,Z(s_{2}) \leq u \mid Y=y) \,\PP_{Y}(\mathrm{d}y)\\
&{}=\int_{\mathcal{F}}  \PP(y(s_1) Z_{y}(s_1) \leq u,y(s_2) Z_{y}(s_2) \leq u)\,\PP_{Y}(\mathrm{d}y)\\
&{}=\int_{\mathcal{F}} \PP\left(Z_{y}(s_1)\leq \frac{u}{y(s_1)},Z_{y}(s_2) \leq \frac{u}{y(s_2)}\right) \,\PP_{Y}(\mathrm{d}y)
\end{align*}
By construction, the joint distribution of $(Z_y(s_1), Z_y(s_2))$ just depends on the values of 
$(y(s_1), y(s_2))$. Define $a_1(y_1,y_2):= \Phi\left(\frac{b}{2}-\frac{1}{b}\log(\frac{y_{2}}{y_{1}}) \right)$ and $a_2(y_1,y_2):= \Phi\left(\frac{b}{2}-\frac{1}{b}\log(\frac{y_{1}}{y_{2}}) \right)$ where $b:=\sqrt{2\gamma(s_1,s_2; y_1, y_2)}$.
Thus, the above expression simplifies in the following way:

\begin{align*}
&\PP(Z(s_{1}) \leq u,Z(s_{2}) \leq u)\\
&{}=\int_{1}^{\infty}\int_{1}^{\infty}\exp\left(- \left[\frac{y_{1}}{u}a_1(y_1,y_2)+\frac{y_{2}}{u}a_2(y_1,y_2)\right] \right) \,\PP_{(Y(s_1),Y(s_2))}(\mathrm{d}y_1,\mathrm{d}y_2)\\
&{}= \int_{1}^{\infty}\int_{1}^{\infty}\exp\left(- \left[\frac{y_{1}}{u}a_1(y_1,y_2)+\frac{y_{2}}{u}a_2(y_1,y_2) \right] \right) y_{1}^{-2} y_{2}^{-2}\mathrm{d}y_{1}\mathrm{d}y_{2}. 
\end{align*}

By a change of variables $y_i = u z_i$, it holds for the tail dependence coefficient
\begin{align*}
 & \chi(s_i,s_j) 
 = \lim_{u  \to \infty} \PP( Z(s_1) > u \mid Z(s_2) > u) \\ 
&=\lim_{u  \to \infty} \frac{ \PP( Z(s_1) > u)+\PP( Z(s_2) > u)-\PP( \max\{Z(s_1),Z(s_2)\} > u)}{\PP(Z(s_2) > u)} \\
&=\lim_{u  \to \infty}\Big(\frac{\int_{1}^{\infty}\int_{1}^{\infty}\left(1-\exp(-\frac{y_{1}}{u})-\exp(-\frac{y_{2}}{u})\right) y_{1}^{-2} y_{2}^{-2}\mathrm{d}y_{1}\mathrm{d}y_{2}}{\int_{1}^{\infty}(1-\exp(-\frac{y_{2}}{u})) y_{2}^{-2}\mathrm{d}y_{2}} \\
&+ \frac{\int_{1}^{\infty}\int_{1}^{\infty}\exp\left(- \Big(\frac{y_{1}}{u}\Phi\left(\frac{b}{2}-\frac{1}{b}\log(\frac{y_{2}}{y_{1}}) \right)+\frac{y_{2}}{u}\Phi\left(\frac{b}{2}-\frac{1}{b}\log(\frac{y_{1}}{y_{2}}) \right)\Big) \right) y_{1}^{-2} y_{2}^{-2}\mathrm{d}y_{1}\mathrm{d}y_{2}}{\int_{1}^{\infty}(1-\exp(-\frac{y_{2}}{u}))y_{2}^{-2}\mathrm{d}y_{2}}\Big) \\
& =\lim_{u  \to \infty} u^{-1} \Big(\frac{\int_{\frac{1}{u}}^{\infty}\int_{\frac{1}{u}}^{\infty}\left(1-\exp(-z_{1})-\exp(-z_{2})\right) (z_{1}z_{2})^{-2} \mathrm{d}z_{1}\mathrm{d}z_{2}}{\int_{\frac{1}{u}}^{\infty}(1-\exp(-z_{2})) z_{2}^{-2} \mathrm{d}z_{2}}  \\
& + \frac{\int_{\frac{1}{u}}^{\infty}\int_{\frac{1}{u}}^{\infty}\exp\left(- \Big(z_{1}\Phi\left(\frac{\tilde{b}}{2}-\frac{1}{\tilde{b}}\log(\frac{z_{2}}{z_{1}}) \right)+z_{2}\Phi\left(\frac{\tilde{b}}{2}-\frac{1}{\tilde{b}}\log(\frac{z_{1}}{z_{2}}) \right)\Big) \right) (z_{1}z_{2})^{-2} \mathrm{d}z_{1}\mathrm{d}z_{2}}{\int_{\frac{1}{u}}^{\infty}(1-\exp(-z)) z^{-2} \mathrm{d}z}\Big),
\end{align*}
where  $\tilde{b}:=\sqrt{2\gamma(s_1,s_2; uz_1, uz_2)}$.

We first consider the denominator and, using that $ 1 - \exp(-z) \geq z - z^2/2$ for all $z \in [0,1] $, we obtain that
\begin{align*}
  & \int_{1/u}^\infty (1 - \exp(-z)) z^{-2} \, \mathrm{d}z
  \geq{} \int_{1/u}^1 (z - z^2/2) z^{-2}\, \mathrm{d}z + \int_{1}^\infty (1 - \exp(-1)) z^{-2} \, \mathrm{d}z \\
  ={}& \log(u)- \frac 1 2 + \frac 1 {2u}+ 1 - \exp(-1)
  ={} \log(u)+ \frac{1}{2u}+ \frac 1 2 - \exp(-1) > \log(u).
\end{align*}
Thus, the denominator is bounded from below by $\log(u)$.

Denote by 
$a_1(u,z_1,z_2):= \Phi\left(\frac{\tilde{b}}{2}-\frac{1}{\tilde{b}}\log(\frac{z_{2}}{z_{1}}) \right)$ and $a_2(u,z_1,z_2):= \Phi\left(\frac{\tilde{b}}{2}-\frac{1}{\tilde{b}}\log(\frac{z_{1}}{z_{2}}) \right)$.
It remains to show that
\begin{align*}
&{} \lim_{u \to \infty} u^{-1} \log(u)^{-1} \Big(\int_{\frac{1}{u}}^{\infty}\int_{\frac{1}{u}}^{\infty}\left(1-\exp(-z_{1})-\exp(-z_{2})\right) (z_{1}z_{2})^{-2} \mathrm{d}z_{1}\mathrm{d}z_{2} \\
&{}+ \int_{\frac{1}{u}}^{\infty}\int_{\frac{1}{u}}^{\infty}\exp\left(- \Big(z_{1}a_1(u,z_1,z_2)+z_{2}a_2(u,z_1,z_2)\Big) \right) (z_{1}z_{2})^{-2} \mathrm{d}z_{1}\mathrm{d}z_{2}=0.
\end{align*}

For the integrands of the double integrals, we use the following result: \\
For $z_1,z_2>0$, we have that
\begin{align*}
    & 1 - \exp(-z_1) - \exp(-z_2) + \exp(-(a_1(u,z_1,z_2)z_1+a_2(u,z_1,z_2)z_2)) \\
={} & (1- \exp(-z_1))(1-\exp(-z_2)) + \int_{a_1(u,z_1,z_2)z_1 +a_2(u,z_1,z_2)z_2}^{z_1+z_2} \exp(-t) \, \mathrm{d}t \\
\leq{}& (1- \exp(-z_1))(1-\exp(-z_2)) \\
   & \quad + \begin{cases}
   1, & z_1 \wedge z_2 \geq 1 \\
   (1-a_1(u,z_1,z_2))z_1 + (1-a_2(u,z_1,z_2))z_2, & z_1 \wedge z_2 < 1 \end{cases} \\
      =:{}& (1- \exp(-z_1))(1-\exp(-z_2)) + g(a_1(u,z_1,z_2),a_2(u,z_1,z_2),z_1,z_2).
\end{align*}

In order to bound the integral over the first summand (the independent case), we may use that 
$1- \exp(-z) \leq \min\{z,1\}$ for $z>0$ and obtain
\begin{align*}
&\lim_{u  \to \infty} u^{-1}\log(u)^{-1}\int_{\frac{1}{u}}^{\infty}\int_{\frac{1}{u}}^{\infty}\left(1-\exp(-z_{1})\right)\left(1-\exp(-z_{2})\right) (z_{1}z_{2})^{-2} \mathrm{d}z_{1}\mathrm{d}z_{2}\\
&{}=\lim_{u  \to \infty} u^{-1}\log(u)^{-1} \Big(\int_{\frac 1 u}^{\infty} (1-\exp(-z))z^{-2}\mathrm{d}z\Big)^{2}\\
&{}\leq \lim_{u  \to \infty} u^{-1}\log(u)^{-1} \Big(\int_{\frac{1}{u}}^{1}z^{-1}\mathrm{d}z+\int_{1}^{\infty}z^{-2}\mathrm{d}z\Big)^{2}\\
&{}=\lim_{u  \to \infty} u^{-1}\log(u)+2u^{-1}+u^{-1}/\log(u)=0.
\end{align*}

We need to split up the second double integral
\begin{align*}
& u^{-1}\log(u)^{-1}\int_{\frac{1}{u}}^{\infty}\int_{\frac{1}{u}}^{\infty}  
  g(a_1(u,z_1,z_2),a_2(u,z_1,z_2),z_1,z_2)
 (z_{1}z_{2})^{-2}\mathrm{d}z_{1}\mathrm{d}z_{2}\\
&{}\leq u^{-1}\log(u)^{-1} \int_{\frac{1}{u}}^{1}\int_{\frac{1}{u}}^{1} [(1-a_1(u,z_1,z_2))z_{1} + (1-a_2(u,z_1,z_2)) z_{2}]
(z_{1}z_{2})^{-2}\mathrm{d}z_{1}\mathrm{d}z_{2}\\
&{} +u^{-1}\log(u)^{-1}\int_{1}^{\infty}\int_{\frac{1}{u}}^{1} [(1-a_1(u,z_1,z_2))z_{1} + (1-a_2(u,z_1,z_2)) z_{2}] (z_{1}z_{2})^{-2}\mathrm{d}z_{1}\mathrm{d}z_{2}\\
&{}+u^{-1}\log(u)^{-1}\int_{1/u}^{1}\int_{1}^{\infty}[(1-a_1(u,z_1,z_2))z_{1} + (1-a_2(u,z_1,z_2)) z_{2}]  (z_{1}z_{2})^{-2}\mathrm{d}z_{1}\mathrm{d}z_{2} \\
&{} + u^{-1}\log(u)^{-1}\int_{1}^{\infty}\int_{1}^{\infty}  (z_{1}z_{2})^{-2}\mathrm{d}z_{1}\mathrm{d}z_{2}
\end{align*}
and find a majorant that is valid for all $u> 1$. To find this, we note that, due to standard probability bounds, for the above choices of $a_1(u,z_1,z_2)$ and $a_2(u,z_1,z_2)$, we have that $a_1(u,z_1,z_2) z_1 + a_2(u,z_1,z_2)z_2 \geq \max\{z_1,z_2\}$ which implies that
\begin{align*}
(1-a_1(u,z_1,z_2)) z_1 + (1-a_2(u,z_1,z_2))z_2 \leq z_1+z_2-\max\{z_1,z_2\} = \min\{z_1,z_2\}. 
\end{align*}

Thus, we get
\begin{align*}
&\lim_{u  \to \infty} u^{-1}\log(u)^{-1} \int_{\frac{1}{u}}^{1}\int_{1}^{\infty} [(1-a_1(u,z_1,z_2))z_{1} + (1-a_2(u,z_1,z_2)) z_{2}]
(z_{1}z_{2})^{-2}\mathrm{d}z_{1}\mathrm{d}z_{2} \\
&\leq{} \lim_{u  \to \infty} u^{-1}\log(u)^{-1} \int_{1/u}^{1}\int_{1}^{\infty} \min\{z_1,z_2\}
(z_{1}z_{2})^{-2} \mathrm{d}z_{1}\mathrm{d}z_{2} \\
&={} \lim_{u  \to \infty} u^{-1}\log(u)^{-1} \int_{1/u}^{1}\int_{1}^{\infty} z_2^{-1} z_1^{-2} \mathrm{d}z_{1}\mathrm{d}z_{2} 
={} \lim_{u  \to \infty} u^{-1} = 0
\end{align*}
and, analogously
\begin{align*}
&{}\lim_{u  \to \infty} u^{-1}\log(u)^{-1} \int_{1}^{\infty} \int_{\frac{1}{u}}^{1} [(1-a_1(u,z_1,z_2))z_{1} + (1-a_2(u,z_1,z_2)) z_{2}] 
(z_{1}z_{2})^{-2}\mathrm{d}z_{1}\mathrm{d}z_{2}\\
&{} \leq{} \lim_{u  \to \infty} u^{-1} = 0. 
\end{align*}
A straightforward computation further reveals that
$$ \lim_{u  \to \infty} u^{-1}\log(u)^{-1}\int_{1}^{\infty}\int_{1}^{\infty} (z_{1}z_{2})^{-2}\mathrm{d}z_{1}\mathrm{d}z_{2} = \lim_{u  \to \infty} u^{-1}\log(u)^{-1} = 0. $$

For the remaining term, we focus on $z_1, z_2 \in (1/u,1)$ and use that
\begin{equation*}
 a_1(u,z_1,z_2)= \Phi\left(\frac{\sqrt{2\gamma(s_1,s_2;uz_1,uz_2)}}{2}-\frac{\log(z_{2}/z_{1})}{\sqrt{2\gamma(s_1,s_2;uz_1,uz_2)}} \right)
\end{equation*}
and 
\begin{equation*}
a_2(u,z_1,z_2)=\Phi\left(\frac{\sqrt{2\gamma(s_1,s_2;uz_1,uz_2)}}{2}-\frac{\log(z_{1}/z_{2})}{\sqrt{2\gamma(s_1,s_2;uz_1,uz_2)}} \right).    
\end{equation*}
For any nonnegative function $f: (1,\infty) \to (0,\infty)$) with
$\lvert z_1-z_2 \rvert > f(u)$, it is
\begin{align*}
   \gamma(s_1,s_2;uz_1,uz_2)\geq c u^{\kappa} f(u)^{\kappa}.  
\end{align*}
Consequently, defining
\begin{equation*}
 c_u := 1- \Phi\left(\frac{\sqrt{2c u^{\kappa} f(u)^{\kappa}}}{2}-\frac{\log(u)}{\sqrt{2c u^{\kappa}f(u)^{\kappa}}} \right) \in (0,1),    
\end{equation*}
we obtain
\begin{equation*}
 1- a_1(u,z_1,z_2) \leq c_u \quad \text{and} \quad 1- a_2(u,z_1,z_2) \leq c_u.   
\end{equation*}
Thus, the remaining double integral can be bounded by
\begin{align*}
& u^{-1}\log(u)^{-1} \int_{\frac{1}{u}}^{1}\int_{\frac{1}{u}}^{1} [(1-a_1(u,z_1,z_2))z_{1} + (1-a_2(u,z_1,z_2)) z_{2}] 
(z_{1}z_{2})^{-2}\mathrm{d}z_{1}\mathrm{d}z_{2} \\
={}&  2 \Big( u^{-1}\log(u)^{-1} \int_{\frac{1}{u}}^{1}\int_{z_2}^{1} [(1-a_1(u,z_1,z_2))z_{1} + (1-a_2(u,z_1,z_2)) z_{2}] 
(z_{1}z_{2})^{-2}\mathrm{d}z_{1}\mathrm{d}z_{2} \Big) \\
\leq{}& 2 \Big( u^{-1}\log(u)^{-1}\int_{\frac{1}{u}}^{1}\int_{z_2}^{z_2+f(u)} \min\{z_1,z_2\} 
(z_{1}z_{2})^{-2}      \mathrm{d}z_{1}\mathrm{d}z_{2} \\
& \quad + c_u u^{-1}\log(u)^{-1}\int_{\frac{1}{u}}^{1}\int_{z_2+f(u)}^{1}  (z_1+z_2) 
(z_{1}z_{2})^{-2}     \mathrm{d}z_{1}\mathrm{d}z_{2} \Big)
\end{align*}

Now, set $f(u)=u^{-1}\log(u)^{\beta}$ with $ 1/\kappa < \beta < 1$. Then, for the first term, it holds
 \begin{align*}
 & \lim_{u  \to \infty} 2 \Big( u^{-1}\log(u)^{-1}\int_{\frac{1}{u}}^{1}\int_{z_2}^{z_2+f(u)} \min\{z_1,z_2\} 
(z_{1}z_{2})^{-2}      \mathrm{d}z_{1}\mathrm{d}z_{2} \Big) \\
 &{}= \lim_{u  \to \infty} 2 \Big( u^{-1}\log(u)^{-1}\int_{\frac{1}{u}}^{1} z_{2}^{-1}
 \int_{z_2}^{z_2+f(u)} z_{1}^{-2}  \mathrm{d}z_{1}\mathrm{d}z_{2} \Big) \\
&{}\leq \lim_{u  \to \infty} 2 \Big( u^{-1}\log(u)^{-1}\int_{\frac{1}{u}}^{1}
f(u) z_{2}^{-3}\mathrm{d}z_{2} \Big) \\
&{}= \lim_{u  \to \infty} u^{-1}\log(u)^{-1}f(u)(u^2-1)=0
 \end{align*}
since $\lim_{u \to \infty} f(u)u/\log(u)=\lim_{u \to \infty} \log(u)^{\beta-1} = 0 $.
For the second term, it is
\begin{align*}
 & \lim_{u  \to \infty} 2 c_u u^{-1}\log(u)^{-1} \Big( \int_{\frac{1}{u}}^{1}\int_{z_2+f(u)}^{1} (z_1+z_2) 
(z_{1}z_{2})^{-2}     \mathrm{d}z_{1}\mathrm{d}z_{2} \Big) \\
&{}\leq   \lim_{u  \to \infty} 2 c_u u^{-1}\log(u)^{-1}\Big( \int_{\frac{1}{u}}^{1}\int_{z_2}^{1} \Big(z_1^{-1}
z_{2}^{-2}+z_2^{-1}
z_{1}^{-2}\Big)  \mathrm{d}z_{1}\mathrm{d}z_{2} \Big) \\
&{}= \lim_{u  \to \infty} 2c_u u^{-1}\log(u)^{-1} \Big( \int_{\frac{1}{u}}^{1}  \Big(-z_{2}^{-1}-z_{2}^{-2}\log(z_2)+z_{2}^{-2}\Big)  \mathrm{d}z_{2} \Big) \\
&{}= \lim_{u  \to \infty} 2c_u u^{-1}\log(u)^{-1} \Big(\log(1/u)+u-\frac{\log(1/u)+1}{1/u}\Big)
{}= \lim_{u  \to \infty} 2c_u = 0
\end{align*}
because
\begin{align*}
&{}\lim_{u \to \infty} \sqrt{ c u^{\kappa} f(u)^{\kappa}}-\frac{\log(u)}{\sqrt{c u^{\kappa} f(u)^{\kappa}}}
=  \lim_{u \to \infty} \sqrt{ c \log(u)^{\beta\kappa} }-\frac{\log(u)^{1-\frac{\beta\kappa}{2}}}{\sqrt{c}}\\
&{}=  \lim_{u \to \infty} c \log(u)^{\frac{\beta\kappa}{2}}-\log(u)^{1-\frac{\beta\kappa}{2}}
= \log(u)^{\frac{\beta\kappa}{2}} \left(c-\log(u)^{1-\beta \kappa}\right)
= \infty.
\end{align*}
\end{proof}

\section{Discussion}
\label{Discussion}
In this paper, we proposed a novel non-stationary approach that can be used for both extremal-$t$ and Brown-Resnick processes, where we include covariates in the corresponding correlation functions and variograms, respectively. As outlined above, our approach can be interpreted as a generalization of the models considered in \cite{blanchetdavison2011}. In contrast to the approach by \citet{huser-genton-16}, our model allows for both bounded and unbounded variograms in the Brown-Resnick case. Compared to the nonparametric methods proposed by \citet{youngman} also for max-stable models, our parametric approach is easier to fit.

In the application, it turned out that the result of the optimization procedure is very sensitive with respect to the initial values. Therefore, we tried different optimization procedures such that, reporting the best parameter values, we obtain reliable results. Thus, the optimization requires careful tuning, which should be investigated further in future research. 

We also took a look at max-stable processes conditional on random covariates and demonstrated that they can result in both asymptotically dependent and asymptotically independent processes as their properties are strongly governed by the tail behaviour and the dependence structure of the covariate process. Consequently, these models are more flexible than classical max-stable models. 
An alternative, more direct approach to obtain non-stationary asymptotically dependent models, could be to use the non-stationary correlation and variogram models in already existing asymptotically independent models.

\subsection*{Acknowledgments}
This work has been supported by the integrated project “Climate
Change and Extreme Events - ClimXtreme Module B - Statistics (subproject B3.1)” funded
by the German Federal Ministry of Education and Research (BMBF) with the grant number 01LP1902I, which is gratefully acknowledged.
In addition, we thank Christopher D\"orr for a helpful suggestion that resulted in the statement of Proposition \ref{prop_schlather}.

\noindent

\bibliographystyle{abbrvnat}
\bibliography{ref.bib}

\end{document}